\RequirePackage[british]{babel}
\documentclass[reqno,a4paper,12pt]{article}

\usepackage{tikz}
\usetikzlibrary{snakes}

\usepackage[latin1]{inputenc}
\usepackage{amsmath,amssymb,amstext,amsthm,amscd}
\usepackage{mathrsfs,eucal,graphicx}
\usepackage{color}
\usepackage{verbatim}
\usepackage{array}
\usepackage[nosort]{cite}
\usepackage[vcentermath,enableskew]{youngtab}
\Yboxdim{4pt}
\usepackage{hyperref}
\newcommand{\fg}{\mathfrak{g}}

\newcommand{\fG}{\mathfrak{G}}

\newcommand{\fL}{\mathfrak{L}}

\theoremstyle{plain}
\newtheorem{lemma}{Lemma}
\newtheorem{proposition}[lemma]{Proposition}

\newtheorem{corollary}[lemma]{Corollary}
\theoremstyle{definition}
\newtheorem{definition}[lemma]{Definition}

\newtheorem{example}[lemma]{Example}
\allowdisplaybreaks
\begin{document}
\renewcommand{\theequation}{\arabic{section}.\arabic{equation}}
\title{\bf $k$-Leibniz algebras from lower order ones:
from Lie triple to Lie $\ell$-ple systems \vskip 1cm}

\author{J. A. de Azc\'{a}rraga$^*$, \\
Dept. Theor. Phys. and IFIC (CSIC-UVEG), \\
Univ. of Valencia, 46100-Burjassot (Valencia), Spain\\
J. M. Izquierdo$^\dagger$,\\
Dept. Theor. Phys., Univ. of Valladolid, \\
47011-Valladolid, Spain    }

{\date{\small{April 2, 2013}}}

\maketitle
\vskip 1cm
\begin{abstract}
Two types of higher order Lie $\ell$-ple systems are introduced in
this paper. They are defined by brackets with $\ell > 3$ arguments
satisfying certain conditions, and generalize the well known Lie triple
systems. One of the generalizations uses a construction that
allows us to associate a $(2n-3)$-Leibniz algebra $\fL$
with a metric $n$-Leibniz algebra $\tilde{\fL}$ by using a $2(n-1)$-linear
Kasymov trace form for $\tilde{\fL}$. Some specific types of $k$-Leibniz
algebras, relevant in the construction, are introduced as well.
Both higher order Lie $\ell$-ple generalizations
reduce to the standard Lie triple systems for $\ell=3$.
\end{abstract}

\vskip 3cm
{\small{${}^*$j.a.de.azcarraga@ific.uv.es , ${}^\dagger$izquierd@fta.uva.es .}}
\newpage

\section{Introduction }
\label{sec:intro}

 It is natural to generalize Lie algebras and
 Poisson structures when their defining brackets have
more than two  arguments. The quest for ($n>2$)-ary algebras of
various types has arisen repeatedly both in mathematics and physics,
there motivated by their possible physical applications; this
was the case of Nambu mechanics \cite{Nambu:73} introduced long ago,
when the quark statistics was being discussed.

A possible higher order Lie algebra structure is
provided by the {\it generalized} or {higher order Lie algebras}
(GLAs) $\mathcal{G}$ \cite{AzPePB:96a,AzBu:96} \cite{Han-Wac:95, Gne:95},
based on an antisymmetric multibracket with even $n$ entries.
The characteristic relation for  $\mathcal{G}$ is the
{\it generalized Jacobi identity} (GJI), which expresses that the total
antisymmetrization of two nested multibrackets vanishes. The GLAs
$\mathcal{G}$ and the GJI reduce to the ordinary Lie algebras $\fg$
and the JI, respectively, when $n=2$. Like Lie algebras, GLAs
have their higher order Poisson counterparts or
{\it generalized Poisson structures} (GPS) \cite{AzPePB:96a}.

Another generalization is provided by the {\it $n$-Lie}
or {\it Filippov algebras}\footnote{There is some unfortunate confusion concerning
the terminology. We call these algebras $n$-Lie (Filippov's
original name) or Filippov algebras indistinctly but {\it not} `Lie $n$-algebras'
(as sometimes the very different FAs and GLAs are {\it both} referred to).
See Sec.~1.1 of \cite{review} for the terminology of the various
$n$-ary algebras. Intermediate generalizations between FAs and
GLAs also exist; see \cite{Vin2:98}. } (FAs) $\fG$ \cite{Filippov,Kas:87,Ling:93}
and, further, by the $n$-Leibniz algebras $\fL$ \cite{Cas-Lod-Pir:02}, which
will be of our concern here.
FAs have a fully skewsymmetric bracket with an arbitrary number
$n\geq 2$ entries, and their characteristic identity is the
{\it Filippov identity} (FI), which expresses that the adjoint action
is a derivation of $\fG$; again, $\fG=\fg$ for $n=2$.
The FAs Poisson counterparts are the Nambu-Poisson
structures (NPS) \cite{Tak:93}, their precedent being $n$=3
Nambu mechanics \cite{Nambu:73}. The NPS also
satisfy the FI (called `five-point
identity' for $n=3$ in \cite{Sah-Val:92} and later `fundamental identity'
in \cite{Tak:93}). The problem of the quantization of
both $n$-ary Poisson structures, GPS and NPS, has been
the subject of intense study and discussion (see
\cite{Sah-Val:92, Tak:93,AIP-B:97,Cu-Za:02} and
the first ref.~in \cite{review} for an outlook and
further references).

Filippov and related algebras have attracted
considerable attention in the last few years
due to their appearance in three-dimensional
superconformal Chern-Simons theories. This
is the case of the metric 3-Lie FAs in the original
$\mathcal{N}=8$ manifestly supersymetric
Bagger-Lambert-Gustavsson \cite{Ba-La:06,Gust:08a} (BLG) model describing
(actually two) coincident membranes, although a $\mathcal{N}=6$
ABJM \cite{Aha-Be-Ja-Mal:08} superconformal Chern-Simons
theory not requiring 3-FAs was soon pointed out.
The key ingredient for the appearance of three-algebras in
BLG models is the close connection between 3-Lie (or 3-Leibniz)
algebras and their associated Lie (gauge) ones, the properties of
which are encoded in the former (eq.~\eqref{assoc} and
Sec.~\ref{subLeib}). The positivity restrictions of the first
$\mathcal{N}=8$ BLG model, that limit its 3-Lie algebra to be
\cite{Pap:08,Ga-Gu:08} a sum of copies of the $A_4$ FA plus central ideals,
led to consider other possibilities (see \cite{Ba-La-Mu-Pa:12}
for a detailed review of developments in the theory of multiple,
parallel membranes in M-theory). The new algebras were
obtained by relaxing the positive definiteness
of the metric \cite{Go-Mi-Ru:08} or the complete antisymmetry of
the FA 3-bracket. Structures of this last type were already known in
the mathematical literature as $n$-Leibniz algebras
\cite{Cas-Lod-Pir:02}; they also satisfy the Filippov identity,
which for an $\fL$ may be called $n$-Leibniz identity
\cite{Cas-Lod-Pir:02}. Relaxing the anticommutativity
led, in particular, to the `real generalized metric
3-algebras' of Cherkis and S\"amann (CS) \cite{Cher-Sa:08}. These
are metric $3$-Leibniz algebras with a `symmetry property'
(Sec.~\ref{subLeib}) and give rise to an $\mathcal{N}=2$ BLG-type
action. Similarly, in the complex case the Bagger and Lambert hermitian
three-algebras \cite{Ba-La:08} are used in their $\mathcal{N}=6$ theory,
which incorporates the ABJM one \cite{Aha-Be-Ja-Mal:08} in a three-algebra approach.

A class of 3-Leibniz algebras is provided by the {\it Lie triple systems}
(Sec.~\ref{triple}), introduced in mathematics
long ago \cite{Jacob:49,Lis:52,Yama:57,Fa:73} (see further
\cite{Ber:00}). They may be defined (Sec.~\ref{triple}) as
$3$-Leibniz algebras with brackets that are antisymmetric in the
first two arguments that, besides the FI, satisfy an additional
cyclic property (there are also Leibniz triple systems \cite{Bre-San:11}
of which Lie triple systems are a particular case, but
these will not be considered here). The aim of this paper is to generalize
Lie triple systems, which have brackets with $\ell$=3 entries,
to higher $\ell\geq 3$; the resulting structure will define
the {\it Lie $\ell$-ple systems}. As a preliminary
example, we recall the connection
between the simple $n=3$ Filippov algebra $A_4$, the
Lie algebra Lie$\,A_4=so(4)$ and its associated triple
system. We will then show that it is possible to obtain an
$\ell$-Leibniz algebra $\fL$ from two Leibniz algebras
$\fL^1$ and $\fL^2$ satisfying certain conditions.
The analysis will lead us to two possible higher order
generalizations of the Lie triple systems, the  Lie $n$-ple and
the Lie $\ell$-ple systems with $\ell=2n-3, n > 3$,
defined by special types of $k$-Leibniz algebras.
For $n$=3=$\ell$, both Lie 3-ple systems coincide
and reproduce the standard Lie triple ones.

The plan of the paper is the following. Sec.~\ref{Fili-Leib}
summarizes the properties of Filippov $\fG$ and $n$-Leibniz
$\fL$ algebras needed here; Sec.~\ref{gen-metr} extends the CS
3-algebras \cite{Cher-Sa:08} to $k=2n-3$. Sec.~\ref{triple}
relates Lie triple systems to a specific type of 3-Leibniz algebras.
Sec.~\ref{mixedcase} provides a method to obtain a metric
$k$-Leibniz algebra, $k=n+m-3$, from two $n$- and $m$-Leibniz
algebras with certain requirements, and considers some
particular cases of later interest. Sec.~\ref{kple}
introduces the two $k$-ple generalizations of the Lie triple
systems and Sec.~\ref{f.r.} contains an outlook.

Only finite dimensional real algebras are considered in this paper.

\section{Filippov $\fG$ and $n$-Leibniz algebras $\fL$}
\label{Fili-Leib}
\subsection{Filippov or $n$-Lie algebras}
\label{subFil}

\begin{definition}
\label{FAdef}
A {\it Filippov} or {\it $n$-Lie algebra} \cite{Filippov} $\fG$ is a vector
space (also denoted by $\fG$) endowed with an antisymmetric
$n$-linear bracket $\fG\times \mathop{\cdots}\limits^n \times \fG \rightarrow \fG$,
$ [X_1,\dots,X_n] \,$,  that obeys the {\it Filippov identity} (FI),
\begin{eqnarray}
\label{FI}
     & &\left[ X_1,\dots, X_{n-1} , [Y_1,\dots , Y_n]\right]
     =\nonumber\\
     & & \quad
     \sum^n_{r=1} \left[ Y_1,\dots, Y_{r-1} , [X_1,\dots, X_{n-1},
     Y_r] , Y_{r+1},\dots , Y_n\right]\, , \, X,Y\in \fG  .
\end{eqnarray}
\end{definition}
\noindent
\\
It is also possible to define a right FI using a right action {\it i.e.},
starting from $[[Y_1,\dots, Y_n ], X_1,\dots , X_{n-1}]$
at the $l.h.s.$ of eq.~\eqref{FI} but, due to the skewsymmetry, the
right and left FIs coincide. Given a basis $\{\mathbf{e}_a\}$
in $\fG$, the $n$-bracket
\begin{equation}
\label{structure}
 [\mathbf{e}_{a_1},\dots,\mathbf{e}_{a_n}]= f_{a_1\dots a_n}{}^b\,
 \mathbf{e}_b \quad ,\; a=1,\dots, \textrm{dim}\,\fG \; ,
\end{equation}
gives  the {\it structure constants} $f_{a_1\dots a_n}{}^b$
of $\fG$, in terms of which the FI reads
\begin{equation}
\label{FIstrconst}
 f_{b_1 \dots b_n}{}^l \; f_{a_1 \dots a_{n-1}l}{\ }^s \,=
 \sum_{k=1}^{n}\,f_{a_1 \dots a_{n-1}b_k}{}^l \; f_{b_1 \dots b_{k-1} l b_{k+1}\dots
 b_n}{}^s \quad ,
\end{equation}

Defining the adjoint action $ad_{\mathscr{X}}\in\mathrm{End}\fG$ of
$\mathscr{X}=(X_1,\dots,X_{n-1})\in \wedge^{n-1}\fG$ by
$ad_{\mathscr{X}} X \equiv \mathscr{X}\cdot X:=
[X_1,\dots,X_{n-1},X]$,
the FI may be rewritten as
\begin{equation}
\label{adjointFI}
  ad_{\mathscr{X}}[Y_1,\dots,Y_n] = \sum^n_{r=1}
  \left[ Y_1,\dots, Y_{r-1} , ad_{\mathscr{X}}
     Y_r , Y_{r+1},\dots , Y_n\right] \; ,
\end{equation}
which expresses that $ad_{\mathscr{X}}$ is an
{\it (inner) derivation} of $\fG$. Clearly,
\begin{equation}
\label{adbasis}
ad_{\mathscr{X}_{a_1 ... a_{n-1}}}\equiv
[\mathbf{e}_{a_1},...,\mathbf{e}_{a_{n-1}},\,\cdot\,]\quad,\quad
ad_{\mathscr{X}_{a_1 ... a_{n-1}}}\, \mathbf{e}_{a_n} = f_{a_1\dots a_n}{}^b\, \mathbf{e}_b \quad .
\end{equation}
A linear transformation $\rho_{\mathscr{X}}$ of a vector space $V$
satisfying the analogue of \eqref{adjointFI} is called
 \cite{Kas:87} a representation $\rho$ of $\fG$ (although
it is really of $\mathscr{X}\in \wedge^{n-1}\fG$); we
shall only consider here $\rho=ad$ (the regular representation).

The $\{ad_\mathscr{X}\}$ are closed under the commutator and
generate the {\it Lie algebra $\mathrm{Lie}\, \fG$
associated with the FA $\fG$}. Indeed, the FI determines the
End$\,\fG$ relation
\begin{equation}
\label{assoc}
 [ad_{\mathscr{X}}, ad_{\mathscr{Y}}] =
 ad_{\mathscr{X}\cdot\mathscr{Y}} \;
 (=- ad_{\mathscr{Y}\cdot\mathscr{X}})
\end{equation}
where, with $\mathscr{Y}=(Y_1,\dots,Y_{n-1})\in \wedge^{n-1}\fG$,
$\mathscr{X}\cdot\mathscr{Y}\in \wedge^{n-1}\fG$ is given by
\begin{equation}
\label{fundam}
      \mathscr{X}\cdot\mathscr{Y}:= \sum_{r=1}^{n-1}
      (Y_1,\dots, Y_{r-1} , ad_{\mathscr{X}}
     Y_r , Y_{r+1},\dots , Y_{n-1}) \quad (\not= -\mathscr{Y}\cdot\mathscr{X}) \quad .
\end{equation}
Eq.~\eqref{fundam} defines the composition law \cite{Gau:96} of
$\mathscr{X}$ and $\mathscr{Y}$ (see also \cite{Filippov,Kas:87}).
It is non-associative since
$\mathscr{X}\cdot(\mathscr{Y}\cdot \mathscr{Z}) -
(\mathscr{X}\cdot \mathscr{Y})\cdot \mathscr{Z} =
\mathscr{Y}\cdot(\mathscr{X}\cdot \mathscr{Z})$, which
follows from the FI. For $\fg$,
this is the JI, $[X,[Y,Z]]-[[X,Y],Z]= [Y,[X,Z]]$
but, since  $\mathscr{X}\cdot\mathscr{Y} \not
= - \mathscr{Y}\cdot\mathscr{X}$ ($cf.$ \eqref{assoc}), the $\mathscr{X}$'s
generate a Leibniz (Sec.~\ref{subLeib}) rather than a Lie algebra.
Due to the importance of the $\mathscr{X}\in \wedge^{n-1}\fG$,
we refer to them as {\it fundamental objects}; in BLG models,
Lie$\,\fG$ is relevant for the gauge transformations.
For $n=2$, Lie$\,\fg\equiv ad\fg$ is generated by
$\{ad_{X_a}\}$ and, obviously, eq.~\eqref{assoc} is
$[ad_X,ad_Y]=ad_{[X,Y]}$; Lie$\,\fg=\fg/Z(\fg)$,
where $Z(\fg) =\{Y|\, ad_X Y=0 \,,\forall X\in\fg\}$
= $\{Y|\, ad_Y X=0 \, \forall X\in\fg \}$ is the centre of $\fg$.

 The study of FAs follows closely (but not
 fully) that of ordinary Lie algebras \cite{Filippov,Kas:87,Ling:93}
(see \cite{review} for further references). For
instance, a FA is simple if $[\fG,\dots,\fG]\not= 0$ and does not
have non-trivial ideals ($I\subset \fG$ is an ideal if
$[\fG,\dots,\fG,Y]\in I\;\forall Y\in I$); semisimple FAs
are direct sums of simple ones. The centre of $\fG$
is the ideal $Z(\fG)=\{Y\in\fG|ad_{\mathscr{X}} Y=0\;\forall \mathscr{X}\}$;
the regular representation is called faithful if $Z(\fG)=0$. Thus,
$ad$ is faithful for the $n$-Lie algebra $\fG/Z(\fG)$ \cite{Kas:87};
if $\fG$ is semisimple, $ad$ is faithful. If it is
further simple, the $ad_{\mathscr{X}}\in\,$Lie$\,\fG$
act irreducibly on $\fG$.

In contrast with the $n=2$ (Lie) case, the only real simple
($n$$>$2)-Lie algebras are $(n+1)$-dimensional \cite{Filippov, Ling:93}
and given by
\begin{equation}
\label{simplefil}
  [\mathbf{e}_{a_1},\dots,\mathbf{e}_{a_n}]= \epsilon_{a_1\dots a_n}{}^b
 \mathbf{e}_b \ ,
\end{equation}
where, in terms of the Levi-Civita symbol, $\epsilon_{a_1\dots a_{n+1}}$,
$\epsilon_{a_1\dots a_n}{}^b = \eta^{ba_{n+1}} \epsilon_{a_1\dots a_n a_{n+1}}$
and $\eta$ is the (euclidean or pseudoeuclidean) metric on the $\fG$ vector space.
The real euclidean simple $n$-Lie algebras are labelled $A_{n+1}$
\cite{Filippov}; the simple lorentzian FAs may be denoted $A_{p+q}$,
with $p+q=n+1$ (these inequivalent real $(n+1)$-dimensional simple
algebras are the same as complex FAs).

There is an analogue of the Cartan criterion that applies to
general $n$-Lie algebras: a FA $\fG$ is semisimple \cite{Kas:95a}
iff the 2($n-1$)-linear {\it Kasymov trace form} $k$ \cite{Kas:87,Kas:95a}
$k:\wedge^{n-1}\fG \times \wedge^{n-1}\fG \rightarrow \mathbb{R}$,
defined for $ad$ as
\begin{equation}
\label{Kasymov}
  k(\mathscr{X},\mathscr{Y}) =k(X_1,\dots,X_{n-1}, Y_1,\dots,Y_{n-1})
  := \textrm{Tr}(ad_{\mathscr{X}} ad_{\mathscr{Y}}) \;,
\end{equation}
is non-degenerate {\it i.e.}, $k(X,\fG,\mathop{\cdots}\limits^{n-2},\fG, \fG,
\mathop{\cdots}\limits^{n-1} ,\fG )=0 \Rightarrow X=0$ (actually,
$k$ may be defined for other representations $\rho$ of $\fG$). For the
regular representation $ad$, $k$ was called \cite{Kas:95a} the Killing
form for $\fG$ (to which $k$ reduces for $n=2$); nevertheless, we
shall refer to $k$ as the Kasymov trace form for $\fG$.
Since $ad_{\mathscr{X}_{a_1 \dots a_{n-1}}}$ is given by the
dim$\fG\times$dim$\fG$ matrix $(ad_{\mathscr{X}_{a_1 \dots a_{n-1}}})_b{}^c$=
$f_{a_1 \dots a_{n-1}b}{}^c$, the coordinates of $k$ are
\begin{equation}
\label{kas-str}
 k_{a_1\dots a_{n-1}b_1\dots b_{n-1}}
 \equiv k(\mathbf{e}_{a_1},\dots , \mathbf{e}_{a_{n-1}},
  \mathbf{e}_{b_1},\dots , \mathbf{e}_{b_{n-1}}) =
  f_{a_1\dots a_{n-1}b}{}^c f_{b_1\dots b_{n-1}c}{}^b \; .
\end{equation}

In general,  $\fG$ is a {\it metric Filippov algebra} when it is
endowed with a non-degenerate bilinear metric $<,>: \fG\times\fG \rightarrow \mathbb{R}$,
$<X,Y> =<Y,X>$, which is Lie$\,\fG$-invariant,
$\mathscr{X}.\eta(Y,Z)\equiv\mathscr{X}.<Y,Z>=0$ {\it i.e.},
\begin{equation}
\label{invariance}
    <ad_{\mathscr{X}} Y_1 , Y_2> + <Y_1, ad_{\mathscr{X}} Y_2> =0 \ , \quad
    \forall Y\in\fG \ ,\ \forall \mathscr{X}\in
    \wedge^{n-1}\fG  \; ;
\end{equation}
we shall refer to \eqref{invariance} as the {\it metricity property}.
If $<\,,\,>$ is euclidean, Lie$\,\fG \subset so(\mathrm{dim}\fG)$;
Lie$\,A_{n+1}=so(n+1)$. For a 3-Lie algebra eq~\eqref{invariance}
simply reads
\begin{equation}
\label{invariance3}
    <[X_1,X_2,Y_1],Y_2> + <Y_1, [X_1,X_2,Y_2]> =0 \;.
\end{equation}
It follows that the structure constants of a metric FA
with all indices down, $f_{a_1\dots a_{n+1}}:=
f_{a_1\dots a_n}{}^b \langle\mathbf{e}_b, \mathbf{e}_{a_{n+1}}\rangle$, are
completely antisymmetric. When $\fG$ is
semisimple it is also metric, Lie$\,\fG$ is semisimple and
the $\fG$ Kasymov trace form may be also
looked at as the non-singular Lie$\,\fG$ Killing metric.

\subsection{$n$-Leibniz algebras}
\label{subLeib}
{\it $n$-Leibniz algebras} \cite{Cas-Lod-Pir:02} $\mathfrak{L}$
result from relaxing the requirement of full skewsymmetry
in Def.~\ref{FAdef}; thus, $n$-FAs are a particular case
of $n$-Leibniz algebras. The FI (or $n$-Leibniz identity
\cite{Cas-Lod-Pir:02} for $\fL$) now depends on whether the
adjoint derivative of the $n$-Leibniz bracket is a left or a right one.
The $n$-Leibniz algebras for which \eqref{FI} is satisfied are
then {\it left} $n$-Leibniz algebras; for definiteness sake,
we shall consider these henceforth.

Since the $n$-Leibniz bracket needs not being anticommutative, the
fundamental objects for $\fL$ are now $\mathscr{X}\in \otimes^{n-1}\fL$.
But, since expressions such as \eqref{assoc} and \eqref{fundam} only
depend on the FI, the composition law \eqref{fundam} for fundamental
objects, $\mathscr{X}\cdot \mathscr{Y}\in \otimes^{n-1}\fL$,
defines again a Leibniz algebra, and there is \cite{Gau:96,Kas:87} still
an associated {\it Lie} algebra Lie$\fL$ (relevant in BLG-type models). Lie$\fL$
is defined \cite{Da-Tak:97} for the quotient space $\otimes^{n-1} \fL / K$,
where $K$ is the kernel of  the adjoint map $ad$,
$K=\{\mathscr{X}\in \otimes^{n-1} \fL\ |\ ad_\mathscr{X}=0\}$
and $ad_\mathscr{X}=0$ obviously means that
$ad_\mathscr{X} Y=[X_1,\dots,X_{n-1},Y]=0 \;\forall Y\in \fL$
(a similar consideration also holds for Lie$\fG$).
An {\it $n$-Leibniz algebra} $\fL$ {\it is metric} when it is endowed
with a Lie$\fL$-invariant scalar product $<\,,\,>$. Then,
condition \eqref{invariance} $\forall\,\mathscr{X} \in \otimes^{n-1}\fL$
is expressed in terms of the structure constants by
$f_{a_1\dots a_{n-1} bc}=- f_{a_1\dots a_{n-1} c b}$.

An example is provided by

\begin{definition}
\label{CSalg}
{\it The `real generalized 3-algebras' (CS 3-algebras)} \cite{Cher-Sa:08}

These are metric (eq.~\eqref{invariance3}) 3-Leibniz
algebras that satisfy the additional `symmetry property'
\begin{equation}
\label{CSsym}
<[X_1,X_2,Y_1],Y_2>=<[Y_1,Y_2,X_1],X_2>  \;,
\end{equation}
which implies $f_{a_1 a_2 b_1 b_2}= f_{b_1 b_2 a_1 a_2}$ for its
structure constants. An obvious particular case
of CS 3-algebras is the simple 3-Lie algebra $A_4$ since
$\langle [\mathbf{e}_{a_1},\mathbf{e}_{a_2},\mathbf{e}_{a_3}],\mathbf{e}_{a_4}\rangle=
\epsilon_{a_1 a_2 a_3 a_4}$ satisfies both eqs.~\eqref{invariance3} and \eqref{CSsym}.
\end{definition}

The symmetry condition \eqref{CSsym}, together with the metricity property
\eqref{invariance3}, implies that the Leibniz 3-bracket of these CS
algebras is antisymmetric in its first two arguments. Similarly, it is
easy to check that when the 3-Leibniz bracket is antisymmetric in
its first two arguments and satisfies the symmetry
property $f_{a_1 a_2 b_1 b_2}= f_{b_1 b_2 a_1 a_2}$ (eq.~\eqref{CSsym})
for a certain metric, the 3-Leibinz algebra is also metric,
$f_{a_1 a_2 b_1 b_2}= -f_{a_1 a_2 b_2 b_1}$
(in which case $f_{a_1 a_2 b_1 b_2}= f_{b_2 b_1 a_2 a_1}$).

\subsection{Higher order Leibniz algebras of CS type}
\label{gen-metr}

The 3-Leibniz algebras of Def.~\ref{CSalg} admit the
following higher order generalization:

\begin{definition}
\label{higherCS}
A `generalized metric $\ell$-algebra' is an
odd euclidean $\ell$-Leibniz algebra defined by an $\ell$-bracket
$[X_1,\dots,X_{n-1},Y_1,\dots,Y_{n-3},Y_{n-2}]$, with $\ell=2n-3$,
$n\geq 3$, that is antisymmetric in the ($n-1$) $X$'s and in ($n-2$)
the $Y$'s and satisfies\\
a) the metricity property \eqref{invariance},
\begin{eqnarray}
\label{ellmetr}
\langle[X_1,\dots,X_{n-1},Y_1,\dots,Y_{n-3},Y_{n-2}], Y_{n-1}\rangle
=\qquad \qquad \qquad \cr
\qquad \qquad \qquad
 - \langle Y_{n-2},[X_1,\dots,X_{n-1},Y_1,\dots,Y_{n-3}, Y_{n-1}] \rangle \; ,
\end{eqnarray}
b) and the symmetry condition
\begin{equation}
\label{defhCS}
<[X_1,\dots,X_{n-1},Y_1,\dots,Y_{n-2}], Y_{n-1}> =
<[Y_1,\dots,Y_{n-1},X_1,\dots,X_{n-2}], X_{n-1}> \; ,
\end{equation}
which reduce to eqs.~\eqref{invariance3}, \eqref{CSsym} for $\ell=3=n$.
\end{definition}
\noindent
The reason of the numbering and the origin of the symmetry property
will become apparent later (eq.~\eqref{def-brack}). Note that, without any
specific assumption and as in the CS $\ell=3$ case above, the symmetry condition \eqref{defhCS}
plus the metricity \eqref{ellmetr} imply ($X_{n-1},X_{n-2}$)-skewsymmetry;
also, ($X_{n-1},X_{n-2}$)-skewsymmetry plus the symmetry condition
\eqref{defhCS} suffice to imply metricity.

   In terms of the structure constants $g_{a_1...a_{n-1}b_1...b_{m-2}}{}^c$
of the $\ell$-Leibniz algebra, the above relations read
\begin{eqnarray}
\label{l-C-S-L}
 g_{a_1...a_{n-1} b_1...b_{n-2} b_{n-1}} &=&
 -g_{a_1...a_{n-1} b_1...b_{n-3} b_{n-1} b_{n-2}} \quad(\mathrm{metricity})\; \\
 g_{a_1...a_{n-1} b_1... b_{n-1}} &=&
 g_{b_1...b_{n-1} a_1...a_{n-1}}  \quad (\mathrm{symm.\, property}) \; .
\end{eqnarray}

\section{Lie triple Systems}
\label{triple}

A particular case of $3$-Leibniz algebras is that of {\it Lie triple
systems} \cite{Jacob:49, Lis:52, Yama:57, Fa:73,Ber:00}. They
have also appeared in physics as {\it e.g.}, in connection with parastatistics
\cite{Okub:94} or the Yang-Baxter equation \cite{Okubo:93,Oku-Kam:96a,Ker:00}.
Further triple (and supertriple) system generalizations may be found
in \cite{Oku-Kam:96b, Oku:03} and references therein.

\begin{definition}
\label{triplesystems}
A {\it Lie triple system} is a (left) $3$-Leibniz
algebra $\mathfrak{L}$ such that its 3-bracket satisfies,
besides the FI,  the conditions
\begin{enumerate}
\item
$[X_1,X_2,Y]=- [X_2,X_1,Y]\quad \forall X_1,X_2,Y\in \mathfrak{L}$
\item
$[X_1,X_2,X_3] + [X_2,X_3,X_1] + [X_3,X_1,X_2] = 0\quad
    \forall X_1,X_2,X_3 \in   \mathfrak{L} $
\end{enumerate}
\noindent
Note that the cyclic property (b), together with (a),
is equivalent to saying that the full antisymmetrization
of the arguments in $[X_1,X_2,X_3]$ vanishes.
\end{definition}

A generic 3-bracket with the property (a) above is a map
$[\;,\,,\,]:\wedge^2 \fL\times \fL \rightarrow \fL$
and hence its symmetry properties correspond to those of
\begin{equation}
\label{decompo3}
 \Yboxdim{12pt}\yng(1,1)\otimes\yng(1) =\yng(1,1,1)\oplus\yng(2,1) \quad .
\end{equation}
Thus, as a $GL(\mathrm{dim}\fL)$-tensor, the irreducible symmetry components
of a 3-bracket antisymmetric in its first two arguments corresponds
to a fully antisymmetric one or to a 3-bracket with the
mixed symmetry of $\yng(2,1)$. When the 3-Leibniz algebra bracket
has the symmetry of $\yng(2,1)$, the cyclic property (b) in
Def.~\ref{triplesystems} is also satisfied and it defines a Lie
triple system. When the $\yng(2,1)$ part is absent, the 3-Leibniz
algebra $\fL$ is actually a FA $\fG$.\\

As mentioned, an euclidean (say) $n$-Leibniz
$\fL$ algebra has an associated {\it Lie} algebra,
Lie$\fL \subset so(\mathrm{dim}\fL)$; thus, the vector space
$\fL$ carries a representation of Lie$\fL$. There is a canonical procedure
\cite{Lis:52,Yama:57,Fa:73,Oku-Kam:96a,Oku-Kam:96b,Oku:03}
that also goes backwards, from a metric $\fg$ to $\fL$.  It starts from
a Lie algebra $\fg$ endowed with a $\fg$-invariant scalar product $(\;,\;)$,
which preserves the euclidean metric $<\,,\,>$ of a vector space $\fL$
($\fg\subset so(\mathrm{dim}\fL)$) on which $\fg$ acts faithfully;
the construction endows $\fL$ with a Leibniz algebra structure.

  Let $\fg=so(4)$ with basis given by the $L_{a_1 a_2}$ that  generate
the rotations in $\fL=\mathbb{R}^4$ vector space. Then, the 3-Leibniz
algebra is defined by
\begin{equation}
\label{gtoL}
(L_{a_1 a_2},L_{b_1 b_2})= <ad_{a_,a_2} \mathbf{e}_{b_1},\mathbf{e}_{b_2}>=
<[\mathbf{e}_{a_1}, \mathbf{e}_{a_2},\mathbf{e}_{b_1}],\mathbf{e}_{b_2}> \; ,
\; a,b=1,\dots, 4 \; .
\end{equation}
Clearly, the symmetry property \eqref{CSsym} is satisfied and
eq.~\eqref{invariance} follows from $L_{b_1 b_2}= -L_{b_2 b_1}$ .
Since the 3-bracket also fulfills the FI (this follows from
Prop.~\ref{n-m-prop} below as a particular case), the construction
leads to CS 3-algebras \cite{Cher-Sa:08},
as shown in \cite{deMed-JMF-Men-Rit:08}, Th.~11
(for the hermitian algebras in \cite{Ba-La:08} see
\cite{deMed-JMF-Men-Rit:08, Palm:09}).

  As a first example of \eqref{gtoL}, let $(\;,\,)$ be the $so(4)$ Killing metric
$\frac{1}{2}\mathrm{Tr}(ad_{a_1 a_2}ad_{a_3 b}) =
\frac{1}{2}\epsilon_{a_1 a_2 b}{}^c \epsilon_{a_3 a_4 c}{}^b$.
This defines the 3-Leibniz algebra
\begin{equation}
\label{cs-so4}
    [\mathbf{e}_{a_1}, \mathbf{e}_{a_2}, \mathbf{e}_{b_1}] =
    -(\delta_{a_1 b_1} \mathbf{e}_{a_2}- \delta_{a_2 b_1} \mathbf{e}_{a_1})=
    -\sum_{\sigma\in S_2}\delta_{a_{\sigma(1)b_1}}\delta_{a_{\sigma(2)}}{}^c \mathbf{e}_c \; .
\end{equation}
Reciprocally, we see that the $ad_{\mathscr{X}_{a_1 a_2}}$ generate
the original $so(4)$ rotations since
\begin{equation}
\label{can-tri}
ad_{\mathscr{X}_{a_1 a_2}} \mathbf{e}_b= [\mathbf{e}_{a_1}, \mathbf{e}_{a_2}, \mathbf{e}_b]
= -(\delta_{a_1 b} \mathbf{e}_{a_2}- \delta_{a_2 b} \mathbf{e}_{a_1})
= L_{a_1a_2} \mathbf{e}_b \; .
\end{equation}
Further, the above 3-bracket does not have a $\yng(1,1,1)$ component. Thus,
it defines a (rather old \cite{Jacob:49}) Lie triple system
which is also an example of the CS 3-algebras
of Def.~\ref{CSalg}. In general, however, Lie triple systems are not metric.

 Let ($\;,\;$) now be the invariant metric
$k^{(2)}(L_{a_1 a_2}, L_{b_1 b_2})=\epsilon_{a_1 a_2 b_1 b_2}$
on $so(4)$, which is symmetric and {\it semi}definite.
Then, eq.~\eqref{gtoL} for $k^{(2)}$ gives
$\epsilon_{a_1 a_2 b_1 b_2} =
\langle [\mathbf{e}_{a_1},\mathbf{e}_{a_2},\mathbf{e}_{b_1}],\mathbf{e}_{b_2}\rangle$,
which corresponds to $\yng(1,1,1)$ and to the $A_4$ 3-Lie algebra,
an obvious example of Def.~\ref{CSalg}. The existence of $k^{(2)}$
is a fortunate accident for the Chern-Simons part of the lagrangian
of the original $A_4$-BLG model: Lie$A_{n+1}=so(n+1)$ is simple
but for $n=3$ and the fully skewsymmetric $k^{(2)}$  does not generalize
to $n>3$ (see the footnote in Sec.~14.4 of \cite{review}). \\

{\it Remark.}$\;$
  The Lie triple system in eq.~\eqref{cs-so4} also follows
from the Kasymov trace form $k$ for the FA $A_4$,
$k:\wedge^2 A_4 \times \wedge^2 A_4 \rightarrow \mathbb{R}$;
since  $A_4$ is simple $k$ is, in fact, the Killing metric for $so(4)$.
Let $<\,,>$ denote the euclidean metric on $\mathbb{R}^4=A_4$. Then,
Kasymov's $k$  ($f_{a_1 a_2 b}{}^c=\epsilon_{a_1 a_2 b}{}^c$ in \eqref{kas-str})
\begin{equation}
\label{cs-so4-eq}
    k_{a_1 a_2 a_3 b} = \frac{1}{2}\mathrm{Tr}(ad_{a_1 a_2}ad_{a_3 b})=
    \frac{1}{2} \epsilon_{a_1 a_2 d}{}^c \epsilon_{a_3 b c}{}^d=
    \langle [\mathbf{e}_{a_1},\mathbf{e}_{a_2},\mathbf{e}_{a_3}],\mathbf{e}_b \rangle
\end{equation}
also reproduces the metric Lie triple system defined by \eqref{cs-so4}.

We shall use the Kasymov form and introduce other `mixed metric'
generalizations below to obtain higher order $k$-Leibniz
algebras from $m$- and $n$-Lie algebras and, later, to introduce
Lie $\ell$-ple systems (Sec.~\ref{kple}).  \\

\section{The $k$-Leibniz algebra associated with two $n$-
and $m$-Leibniz algebras }
\label{mixedcase}

Let $\fL^1$ and $\fL^2$ be, respectively, $n$- and $m$-Leibniz algebras
defined on the same vector space $V$, and let $\fL^2$ be metric with
respect to $\langle \; ,\;\rangle$ so that
$ad^2_{\mathscr{X}}\in so(\mathrm{dim}V)$. Assume now that   \\
\noindent
a) the $m$-bracket of $\fL^2$ satisfies
\begin{equation}
\label{defhCS3}
<[Z_1,\dots,Z_{m-3},X_1,X_2,Y_1]_{\fL^2}, Y_2> =
<[Z_1,\dots ,Z_{m-3},Y_1,Y_2,X_1]_{\fL^2}, X_2> \; ,
\end{equation}
(a condition satisfied by all metric $m$-FAs and that, for $m=3$,
reduces to the symmetry property eq.~\eqref{CSsym}). Further,\\
\noindent
b)
$ad_\mathscr{Y}^2$, $\mathscr{Y}\in \otimes^{m-1}\fL^2$,  is a
derivation of $\fL^1$,
\begin{equation}
\label{adjointFI2}
  ad_{\mathscr{Y}}^2 [X_1,\dots,X_n]_{\fL^1} =
  \sum^n_{r=1} \left[ X_1,\dots, ad_{\mathscr{Y}}^2
     X_r ,\dots X_n\right]_{\fL^1} \; ,
\end{equation}
where $ad_{\mathscr{X}}^2\, X_r\,$ is the $m$-bracket in
$\fL^2$. Then, the following proposition follows:

\begin{proposition}
\label{n-m-prop}
Let $\fL^1$ and $\fL^2$ be as above satisfying
conditions \eqref{defhCS3}, \eqref{adjointFI2}. Let the
generalization of the Kasymov trace form
$k:(\mathscr{X}, \mathscr{Y}) \rightarrow \mathbb{R}$, where now
$\mathscr{X}\in \otimes^{n-1}\fL^1\,,\, \mathscr{Y}\in \otimes^{m-1}\fL^2$,
be defined by $k(\mathscr{X}, \mathscr{Y})=
\mathrm{Tr}(ad^1_\mathscr{X} ad^2_\mathscr{Y})= k(\mathscr{Y}, \mathscr{X})$;
clearly, it reduces to Kasymov's $k$ for $\fL^1 = \fL^2$. Then,
the ($n+m-3$)-bracket defined on $V$ by
\begin{equation}
\label{n-m-prop-eq}
\langle [X_1,\dots,X_{n-1},Y_1,\dots,Y_{m-2}]_{\fL}, Y_{m-1}\rangle =
\mathrm{Tr}(ad_{X_1\dots X_{n-1}}^1 ad_{Y_1 \dots Y_{m-1}}^2 ) \; ,
\end{equation}
satisfies the FI and therefore defines a $k$-Leibniz
algebra $\fL$, $k=n+m-3$. Moreover, this $k$-Leibniz algebra is
metric $w.r.t.$ $\langle \, ,\,\rangle$.
\end{proposition}
\begin{proof}

  Let $\{\mathbf{e}_a \}$ be a basis of the common underlying
vector space $V$ of the $\fL^1$ and $\fL^2$ algebras
and let $f_{a_1\cdots a_n}{}^b$ and $h_{b_1\cdots b_m}{}^c$ be,
respectively, their structure constants in that basis,
\begin{equation}
\label{strcons}
   \left[\mathbf{e}_{a_1},\cdots ,\mathbf{e}_{a_n}\right]_{\fL^1} =
   f_{a_1\cdots a_n}{}^b \mathbf{e}_b \quad, \quad
    \left[\mathbf{e}_{b_1},\cdots ,\mathbf{e}_{b_m}\right]_{\fL^2} =
    h_{b_1\cdots b_m}{}^c \mathbf{e}_c \; ,
\end{equation}
where, since $\fL^2$ is metric,
\begin{equation}
\label{metL2}
  h_{b_1\cdots b_{m-1}u_1u_2} = - h_{b_1\cdots b_{m-1}u_2u_1} \; .
\end{equation}
It follows from eq.~\eqref{n-m-prop-eq} that the structure constants
$g_{a_1\cdots a_{n-1} b_1\cdots b_{m-2}}{}^d$
\begin{equation}
\label{Leib-str}
     [\mathbf{e}_{a_1},\cdots , \mathbf{e}_{a_{n-1}},
 \mathbf{e}_{b_1},\cdots , \mathbf{e}_{b_{m-2}}]_{\fL}=
 g_{a_1\cdots a_{n-1}b_1\cdots b_{m-2}}{}^d \mathbf{e}_d
\end{equation}
of the $(n+m-3)$-bracket defining $\fL$ are expressed in terms of those of
$\fL^1$ and $\fL^2$ as
\begin{equation}
\label{gf2}
   g_{a_1\cdots a_{n-1}b_1\cdots b_{m-2}d} =
f_{a_1\cdots a_{n-1}}{}^{uv} h_{b_1\cdots b_{m-2}dvu} \; ,
\end{equation}
where indices are raised and lowered by the metric
$\langle\; ,\;\rangle$ on $V$.

To prove that the $k=(n-m-3)$-bracket \eqref{n-m-prop-eq} defines a
$k$-Leibniz algebra $\fL$ it suffices to check the FI
(eq.~\eqref{FIstrconst}),
\begin{eqnarray}
\label{LI-comp}
 & & g_{a_1\cdots a_{n-1}b_1\cdots b_{m-2}}{}^l g_{c_1\cdots c_{n-1}
d_1\cdots d_{m-3}l}{}^s  \nonumber\\
& & -\sum_{r=1}^{n-1} g_{c_1\cdots c_{n-1}d_1\cdots d_{m-3}
a_r}{}^l g_{a_1\cdots a_{r-1}l a_{r+1}\cdot a_{n-1} b_1\cdots
b_{m-2}}{}^s
\nonumber\\
& & - \sum_{r=1}^{m-2} g_{c_1\cdots c_{n-1}d_1\cdots d_{m-3}
b_r}{}^l g_{a_1\cdots a_{n-1} b_1\cdots b_{r-1} l b_{r+1} \cdots
b_{m-2}}{}^s =0 \; .
\end{eqnarray}
The FI for $\fL^2$ and the derivation property \eqref{adjointFI2}
read, respectively,
\begin{eqnarray}
\label{FI-coord}
   h_{b_1\cdots b_m}{}^l h_{a_1\cdots a_{m-1}l}{}^s &=& \sum^m_{r=1}
h_{a_1\cdots a_{m-1}b_r}{}^l h_{b_1\dots b_{r-1} l b_{r+1} \cdots b_m}{}^s \; ,
\nonumber\\
  f_{a_1\cdots a_n}{}^l h_{b_1\cdots b_{m-1}l}{}^s &=& \sum^n_{r=1}
h_{b_1\cdots b_{m-1}a_r}{}^l f_{a_1\dots a_{r-1} l a_{r+1} \cdots a_n}{}^s \; .
\end{eqnarray}
Using eq.~\eqref{gf2} to express the $g$'s in terms of the
$f$'s and $h$'s in the FI \eqref{LI-comp}
for $\fL$, and property \eqref{defhCS3} of $\fL^2$,
\begin{equation}
\label{sim-in-coord}
      h_{b_1\cdots b_{m-3}u_1u_2v_1v_2}=
       h_{b_1\cdots b_{m-3}v_1v_2u_1u_2}\; ,
\end{equation}
the $l.h.s.$ of \eqref{LI-comp} becomes
\begin{eqnarray}
\label{LI-subst}
 & & f_{a_1\cdots a_{n-1}uv} h_{b_1\cdots b_{m-2}}{}^{lvu}
f_{c_1\cdots c_{n-1}}{}^{wt} h_{d_1\cdots d_{m-3}l}{}^s{}_{tw}
\nonumber\\
& & - f_{c_1\cdots c_{n-1}uv}h_{b_1\cdots b_{m-2}}{}^s{}_{tw}
\sum_{r=1}^{n-1} h_{d_1\cdots d_{m-3}vua_r}{}^{l}
f_{a_1\cdots a_{r-1}la_{r+1}\cdots a_{n-1}}{}^{wt} \nonumber\\
& & - f_{c_1\cdots c_{n-1}uv} f_{a_1\cdots a_{n-1}}{}^{wt}
\sum_{r=1}^{m-2} h_{d_1\cdots d_{m-3}vub_r}{}^{l} h_{b_1\cdots
b_{r-1}lb_{r+1}\cdots b_{m-2}}{}^s{}_{tw}\ .
\end{eqnarray}
Now, using eqs.\eqref{FI-coord}, the sums in \eqref{LI-subst}
can be rewritten as
\begin{eqnarray}
\label{sum1}
  & & - \sum_{r=1}^{n-1} h_{d_1\cdots d_{m-3}vua_r}{}^{l}
f_{a_1\cdots a_{r-1}la_{r+1}\cdots a_{n+1}}{}^{wt} \nonumber\\
 & & = -f_{a_1\cdots a_{n-1}}{}^{wl}
h_{d_1\cdots d_{m-3}vul}{}^t + h_{d_1\cdots
d_{m-3}vu}{}^{wl} f_{a_1\cdots a_{n-1}l}{}^t\ ,
\end{eqnarray}
and
\begin{eqnarray}
\label{sum2}
 & & -\sum_{r=1}^{m-2} h_{d_1\cdots d_{m-3}b_r}{}^{lvu}
h_{b_1\cdots b_{r-1}lb_{r+1}\cdots b_{m-2}}{}^s{}_{tw}
= - h_{b_1\cdots b_{m-2}}{}^s{}_t{}^l
h_{d_1\cdots d_{m-3}}{}^{vu}{}_{lw} \nonumber\\
& & + h_{d_1\cdots d_{m-3}}{}^{vusl}
h_{b_1\cdots b_{m-2}ltw}
+ h_{d_1\cdots d_{m-3}}{}^{vu}{}_t{}^l h_{b_1\cdots
b_{m-2}}{}^s{}_{lw}\ .
\end{eqnarray}
Inserting \eqref{sum1} and \eqref{sum2} into \eqref{LI-subst}
and using again property \eqref{defhCS3},
it is found that \eqref{LI-subst} vanishes. Hence, the FI
\eqref{LI-comp} is satisfied and the $k$-bracket
(eqs.~\eqref{n-m-prop-eq}, \eqref{Leib-str})
defines a $k$-Leibniz algebra $\fL$ associated with $\fL^1$
and $\fL^2$. Further, $\fL$  is metric with respect to $\langle\; ,\; \rangle$,
as can be easily seen from eq. \eqref{gf2} by using the assumed metricity of
$\fL^2$ (eq.~\eqref{metL2}) together with property \eqref{sim-in-coord}.
\end{proof}

We shall use below two particular Corolaries of Prop.~\ref{n-m-prop}.

\begin{corollary}
\label{corcinco}
Let $\tilde{\fL} $ be a metric $n$-Leibniz algebra with an
$n$-bracket skewsymmetric in its first $n-1$ arguments that satisfies
condition \eqref{defhCS3}. Obviously, condition \eqref{adjointFI2} holds.
Then, the $\ell$-bracket with $\ell=2n-3$, defined by
eq.~\eqref{n-m-prop-eq} with $\fL^1=\fL^2=\tilde{\fL}$,
\begin{equation}
\label{def-brack}
\langle [ X_1,\cdots ,X_{n-1}, Y_1,\cdots ,Y_{n-2}]_{\fL}, Y_{n-1}\rangle =
\textrm{Tr}(ad_{(X_1,\cdots ,X_{n-1})} ad_{(Y_1,\cdots ,Y_{n-1})}) \; ,
\end{equation}
where $\langle \; ,\;\rangle$ is the invariant metric on the
common vector space,  defines a metric $\ell$-Leibniz algebra $\fL$
with a bracket that is antisymmetric under both the first $n-1$ and last $n-2$
arguments. Further, it satisfies \eqref{defhCS} by construction and hence
$\fL$ is a generalized metric $\ell$-algebra in the sense of Def.~\ref{higherCS}.
\end{corollary}
\noindent
Clearly, the conditions in this Corollary are met when $\tilde{\fL}$ is
in particular a metric $n$-Lie algebra. This is the case of

\begin{example}
\label{ht}
Let $\tilde{\fL}=A_{n+1}$ and $<\;,\;>$ euclidean in Cor.~\ref{corcinco}.
Then, eqs.~\eqref{adbasis} and
\eqref{simplefil} give
\begin{equation}
\label{son+1ac}
 ad_{a_1 \dots a_{n-1}} . \mathbf{e}_{a_n} =
 \epsilon_{a_1\dots a_{n-1} a_n}{}^{a_{n+1}} \mathbf{e}_{a_{n+1}} \; , \; a=1,\dots,n+1
\end{equation}
and, by eq.~\eqref{def-brack}, the Kasymov trace form leads to
\begin{eqnarray}
\label{2n-3-brack}
 k_{a_1,\dots,a_{n-1},b_1,\dots,b_{n-1}} &=& \frac{1}{2}\mathrm{Tr}(ad_{a_1,\dots,a_{n-1}}ad_{b_1,\dots,b_{n-1}}) =\cr
 \frac{1}{2}\epsilon_{a_1...a_{n-1} b}{}^c \,\epsilon_{b_1...b_{n-1} c}{}^b  &=&
-\sum_{\sigma\in S_{n-1}} \delta_{a_{\sigma(1)b_1}}\dots \delta_{a_{\sigma(n-2)b_{n-2}}}
\delta_{a_{\sigma(n-1)b_{n-1}}}  \cr
&=&\langle[\mathbf{e}_{a_1},\dots,\mathbf{e}_{a_{n-1}},
\mathbf{e}_{b_1},\dots, \mathbf{e}_{b_{n-2}}] , \mathbf{e}_{b_{n-1}}\rangle \; ,
\end{eqnarray}
which defines a ($2n-3)$-bracket antisymmetric in its first ($n-1$) and
second ($n-2$) indices separately. This ($2n-3$)-Leibniz algebra will
be used to define the Lie $\ell$-ple system in Sec.~\ref{L-l-ple}.

It is sufficient to replace $\delta_{ab}$ by the Minkowskian $\eta_{ab}$
to account for the case of the Lorentzian algebras $A_{p+q}$.
Notice that we may also follow the procedure in Sec.~\ref{triple}
(eq.~\eqref{gtoL}) to obtain the above ($\ell=2n-3$)-Leibniz
algebra $\fL$ from $\fg=so(n+1)$ and its Killing metric.
Characterizing the $so(n+1)$ generators by $n-1$
indices\footnote{The action of $so(n+1)$ on $\mathbb{R}^{n+1}$ is given
({\it cf.}~eq.~\eqref{can-tri}) by eq.~\eqref{son+1ac}. This leads,
using the dual as
$L_{b_1 b_2}= \frac{-1}{(n-1)!} \epsilon_{b_1 b_2}{}^{a_1\dots a_{n-1}} ad_{a_1 \dots a_{n-1}}$,
to the familiar expression
$L_{b_1 b_2}  \mathbf{e}_{a_n}$=$\frac{-1}{(n-1)!} \epsilon_{b_1 b_2}{}^{a_1\dots a_{n-1}}
\epsilon_{a_1\dots a_{n-1} a_n}{}^{a_{n+1}} \mathbf{e}_{a_{n+1}} =
-(\delta_{b_1 a_n} \mathbf{e}_{b_2} - \delta_{b_2 a_n} \mathbf{e}_{b_1}) \; $.  }
($a_1,\dots, a_{n-1}$)
(${n+1\choose n-1}={{n+1}\choose 2}$), the Killing metric
$(\;,\;)$ on $so(n+1)$ leads to the $k$ in \eqref{2n-3-brack} and
the $\ell$-Leibniz algebra defined there.
\end{example}
\noindent

\begin{corollary}
\label{corseis}
Let $\fL^2 $ be a CS $3$-algebra (Def.~\ref{CSalg}). Thus,
condition \eqref{defhCS3} holds. Let $\fL^1$ be an
$n$-Leibniz algebra on the same vector space $V$ endowed
with an $n$-bracket skewsymmetric in its first $n-1$ arguments
and let $ad^2_{(X_1,X_2)}$ be a derivation of $\fL^1$. Then,
the $n$-bracket ($h=n+3-3$)
\begin{equation}
\label{def-brack-new}
\langle [ X_1,\cdots ,X_{n-1}, X_n], Y \rangle =
\textrm{Tr}(ad^1_{(X_1,\cdots ,X_{n-1})} ad^2_{(X_n,Y)})
\end{equation}
is skewsymmetric under the interchange of its first $n-1$ arguments
and defines by Prop.~\ref{n-m-prop} a metric
$n$-Leibniz algebra  $\fL$.
\end{corollary}

For $n=3$, the metric 3-Filippov algebras obtained from
Cor.~\ref{corcinco} and \ref{corseis} are both
CS 3-algebras.

\begin{example}
\label{excorseis}
Consider first a metric $m$-Leibniz algebra $\fL^2$ defined
on the ($n+1$)- dimensional space of $\fL^1=\fG^1 = A_{n+1}$.
It follows that the adjoint action of $\fL^2$ is a derivation of $A_{n+1}$
{\it i.e.}, the second equation in \eqref{FI-coord} for
$f_{a_1\cdots a_n}{}^l=\epsilon_{a_1\cdots a_n}{}^l$,
\begin{equation}
\label{FI-coord2}
  \epsilon_{a_1\cdots a_n}{}^l h_{b_1\cdots b_{m-1}l}{}^s = \sum^n_{r=1}
h_{b_1\cdots b_{m-1}a_r}{}^l \epsilon_{a_1\dots a_{r-1} l a_{r+1} \cdots a_n}{}^r \; ,
\end{equation}
holds. To see it, consider the Schouten-type identity
\begin{equation}
\label{schout}
    h_{b_1\cdots b_{m-1}}{}^l{}_{[s} \epsilon_{a_1\cdots a_nl]} \equiv  0
    \quad,\quad a,b,l,s=1\dots,n+1   \; .
\end{equation}
This reduces to the sum of the $n+2$ cyclic permutations
\begin{equation}
\label{schout2}
    h_{b_1\cdots b_{m-1}}{}^l{}_{s} \epsilon_{a_1\cdots a_nl} = \sum^n_{r=1}
    h_{b_1\cdots b_{m-1}}{}^l{}_{a_r} \epsilon_{a_1\cdots a_{r-1} s a_{r+1}\cdots a_nl}
    +h_{b_1\cdots b_{m-1}}{}^l{}_{l} \epsilon_{a_1\cdots a_ns}\ .
\end{equation}
Since the last term vanishes by the complete antisymmetry of the $h$'s
with all indices down (recall that $\fL^2$ is metric),
what remains reproduces \eqref{FI-coord2}.

Now, let $\fL^2$ be a metric CS-Lie algebra. Then, all the conditions
of Cor.~\ref{corseis} are met and eq.~\eqref{def-brack-new} defines
an $n$-Leibniz algebra $\fL$ which is skewsymmetric in its
first $n-1$ arguments.
\end{example}

\section{Higher order Lie $k$-ple systems}
\label{kple}

\subsection{Lie $n$-ple systems: a first generalization}

There is a higher-order generalization of the Lie triple system
that is very close to Def.~\ref{triplesystems}. For it, it is sufficient to
look at the symmetry pattern of an $n$-Leibniz
algebra with skewsymmetric fundamental objects.
The symmetry pattern decomposition is determined by
\begin{equation}
\label{Lie n-ple}
\raisebox{-0.6cm}{\begin{tikzpicture}[scale=0.3]
  \draw[snake=brace] (-0.2,-1) -- (-0.2,3.5);
  \draw (0,-1) rectangle (1,3.5);
  \draw (0,2.5) -- (1,2.5);
  \draw (0,0) -- (1,0);
  \node at (-1.3,1.2) {\mbox{\scriptsize $n$-$1$}};
  \node at (0.5,1.5) {$\cdot$};
  \node at (0.5,2) {$\cdot$};
  \node at (0.5,1) {$\cdot$};
  \node at (0.5,0.5) {$\cdot$};
\end{tikzpicture}} \otimes
\raisebox{-0.05cm}{\begin{tikzpicture}[scale=0.3]
\draw (0,-1) rectangle (1,0);
\end{tikzpicture}}
=
\raisebox{-0.75cm}{\begin{tikzpicture}[scale=0.3]
  \draw[snake=brace] (-0.2,-1) -- (-0.2,4.5);
  \draw (0,-1) rectangle (1,4.5);
  \draw (0,3.5) -- (1,3.5);
  \draw (0,0) -- (1,0);
  \node at (-1,1.8) {\mbox{\scriptsize $n$}};
  \node at (0.5,2.5) {$\cdot$};
  \node at (0.5,2) {$\cdot$};
  \node at (0.5,1.5) {$\cdot$};
  \node at (0.5,2) {$\cdot$};
  \node at (0.5,1) {$\cdot$};
  \node at (0.5,0.5) {$\cdot$};
\end{tikzpicture}}
\ \oplus
\raisebox{-0.55cm}{\begin{tikzpicture}[scale=0.3]
  \draw[snake=brace] (-0.2,-1) -- (-0.2,3.5);
  \draw (1,2.5) rectangle (2,3.5);
  \draw (0,-1) rectangle (1,3.5);
  \draw (0,2.5) -- (1,2.5);
  \draw (0,0) -- (1,0);
  \node at (-1.3,1.2) {\mbox{\scriptsize $n$-$1$}};
  \node at (0.5,1.5) {$\cdot$};
  \node at (0.5,2) {$\cdot$};
  \node at (0.5,1) {$\cdot$};
  \node at (0.5,0.5) {$\cdot$};
\end{tikzpicture}}
\end{equation}
Clearly, the mixed symmetry pattern in the $r.h.s.$ suggests
\begin{definition}
\label{def Lie n-ple}
A {\it Lie $n$-ple system} is given by an $n$-Leibniz algebra
such that its bracket is\\
\noindent
a) skewsymmetric in its first $n-1$ arguments and \\
\noindent
b) satisfies the cyclic property,
\begin{equation}
\label{n-cyclic}
\sum_{cyclic} [X_1,X_2,\dots,X_n]=0 \; .
\end{equation}
\end{definition}
\noindent

\begin{example}
\label{n-ple-facil}
Let  $\fL^1=\fG^1=A_8$ and $\fL^2=\fG^2=A_4\oplus A_4$.
Let the basis of the common vector space $V= \mathbb{R}^8$
be $\{ \mathbf{e}_c \}$, $c,d=1,\dots,8$.
By Ex.~\ref{excorseis}, $\fG^2$ is
a derivation of $\fG^1$. Let the indices of $\fG^2$ be denoted
$a$, $b$ when they refer, respectively, to the
first and second ideals $A_4$; we may set {\it e.g.}
$a=1,\dots,4$ and $b=5,\dots,8$. Then, the structure
constants of $\fG^2$ satisfy $f_{abcd}=0\,\forall c, d$.
Let $\epsilon$ be the $\mathbb{R}^8$ Levi-Civita tensor,
and $\bar{\epsilon}$ that on the $A_4$ ideals.
The structure constants of the $7$-Leibniz algebra $\fL$
constructed as in Cor.~\ref{corseis} are given by
(see eq.~\eqref{gf2})
\begin{equation}
\label{gf3}
  g_{c_1\cdots c_6c_7c_8} = \epsilon_{c_1\cdots c_6}{}^{d_1}{}_{d_2}
  f_{c_7c_8}{}^{d_2}{}_{d_1} \; ;
\end{equation}
they are antisymmetric in $c_1\dots c_6$ and $c_7 c_8$ separately.
Since $c_7$ and $c_8$ cannot take values in different
ideals without $f_{c_7c_8}{}^{d_2}{}_{d_1}$ being zero, we are
left with two possibilities $g_{c_1\cdots c_6a_1a_2}$ and
$g_{c_1\cdots c_6b_1b_2}$. Let us consider the first case
so that $f$ has only indices in the first $A_4$,
\begin{equation}
\label{gf4}
  g_{c_1\cdots c_6a_1a_2} = - \epsilon_{c_1\cdots c_6a_3a_4}
  f_{a_1a_2}{}^{a_3a_4} \; .
\end{equation}
It is clear that among the $c_1\dots c_6$ of the above expression
there must be four $b$'s and two $a$'2 in different orders,
so $g_{c_1\cdots c_6a_1a_2}$ may be written as
\begin{eqnarray}
\label{gf5}
  g_{c_1\cdots c_6a_1a_2} &\sim&
  \delta^{b_1\cdots b_4 a_5a_6}_{c_1\cdots \cdots c_6}
  \bar{\epsilon}_{b_1\cdots b_4}  \bar{\epsilon}_{a_5a_6a_3a_4}
  \bar{\epsilon}_{a_1a_2}{}^{a_4a_3}\nonumber\\
  &\sim & \delta^{b_1\cdots b_4 a_5a_6}_{c_1\cdots \cdots c_6}
  \bar{\epsilon}_{b_1\cdots b_4} \delta_{a_5a_1}\delta_{a_6a_2} \; ;
\end{eqnarray}
there is a similar expression for $g_{c_1\cdots c_6b_1b_2}$,
obtained by substituting $b$'s for  $a$'s.

We may check that the $7$-Leibniz algebra $\fL$
determined by eq.~\eqref{gf3} satisfies condition
\eqref{n-cyclic}, which is equivalent to requiring that
the full antisymmetrization of its $7$-bracket vanishes.
Actually, this is always the case for $n$-Leibniz algebras
obtained as in Ex.~\ref{excorseis} when $\fL^2=\fG^2$ is a
metric 3-Filippov algebra. Indeed, the structure constants of
the $n$-Leibniz algebra $\fL$ are given by
\begin{equation}
\label{ex9a}
    g_{a_1\cdots a_{n-1}a_na_{n+1}}=\epsilon_{a_1\cdots a_{n-1}uv}
    f_{a_na_{n+1}}{}^{vu}\ ,
\end{equation}
where the $f$'s are the structure constants of the 3-Lie algebra $\fG^2$. The full
antisymmetrization of the $n$ entries of the bracket of $\fL$ corresponds to
\begin{equation}\label{ex9b}
  g_{[a_1\cdots a_{n-1}a_n]a_{n+1}}=\epsilon_{uv[a_1\cdots a_{n-1}}
  f_{a_n]a_{n+1}}{}^{vu}\ .
\end{equation}
But this is zero, as can be seen by dualizing the $r.h.s.$ of
\eqref{ex9b}, which gives
\begin{equation}
\label{ex9c}
    \epsilon_{uv a_1\cdots a_{n-1} } f_{a_na_{n+1}}{}^{vu}
    \epsilon^{a_1\cdots a_nb}= (n-1)! \delta^{a_n b}_{uv}
  f_{a_na_{n+1}}{}^{vu}= 0\ ,
\end{equation}
the last equality being a consequence of the complete antisymmetry
of the structure constants of $\fG^2$.

This last fact allows us to construct other examples of Lie
$n$-ple systems based on a simple $n$-Lie algebra $\fG^1=A_{n+1}$ and
a metric 3-Lie algebra $\fG^2$.

%

\end{example}

\subsection{Lie $\ell$-ple systems, $\ell=2n-3$}
\label{L-l-ple}

This second generalization uses $\ell$-Leibniz algebras with
and $\ell$-bracket antisymmetric in its
first ($n-1$) and last ($n-2$) arguments (as in Cor.~\ref{corcinco}).
To introduce the $\ell$-ple Lie systems we have to look
for the property that replaces (b) in Def.~\ref{triplesystems}
when $\ell$=($2n-3$) or, equivalently, for the symmetry pattern
of the generic $\ell$-Leibniz bracket that generalizes
$\Yboxdim{8pt}\yng(2,1)$ in eq.~\eqref{Leib-str} when $\ell>3$ and
reduces to it for $\ell=3$.

Let $\ell=(2n-3)>3$. An $\ell$-bracket skewsymmetric in its
first $n-1$ and last $n-2$ entries has the generic symmetry of
\raisebox{-0.5cm}{\begin{tikzpicture}[scale=0.3]
  \draw[snake=brace] (-0.2,-1) -- (-0.2,3.5);
  \draw (0,-1) rectangle (1,3.5);
  \draw (0,2.5) -- (1,2.5);
  \draw (0,0) -- (1,0);
  \node at (-1.3,1.2) {\mbox{\scriptsize $n$-$1$}};
  \node at (0.5,1.5) {$\cdot$};
  \node at (0.5,2) {$\cdot$};
  \node at (0.5,1) {$\cdot$};
  \node at (0.5,0.5) {$\cdot$};
\end{tikzpicture}}
$\otimes$
\raisebox{-0.4cm}{\begin{tikzpicture}[scale=0.3]
  \draw (1,0) rectangle (2,3.5);
  \draw (1,2.5) -- (2,2.5);
  \draw (1,2.5) -- (1,2.5);
  \draw (1,1) -- (2,1);
  \draw[snake=brace] (2.2,3.5) -- (2.2,0);
  \node at (1.5,1.5) {$\cdot$};
  \node at (1.5,2) {$\cdot$};
  \node at (3.3,1.8) {\mbox{\scriptsize $n$-$2$}};
\end{tikzpicture}}.
Its decomposition in terms of irreducible Young patterns is
given by
\begin{equation}
\label{decomp}
   \raisebox{-0.5cm}{\begin{tikzpicture}[scale=0.3]
  \draw[snake=brace] (-0.2,-1) -- (-0.2,3.5);
  \draw (0,-1) rectangle (1,3.5);
  \draw (0,2.5) -- (1,2.5);
  \draw (0,0) -- (1,0);
  \node at (-1.3,1.2) {\mbox{\scriptsize $n$-$1$}};
  \node at (0.5,1.5) {$\cdot$};
  \node at (0.5,2) {$\cdot$};
  \node at (0.5,1) {$\cdot$};
  \node at (0.5,0.5) {$\cdot$};
\end{tikzpicture}}
\otimes
\raisebox{-0.4cm}{\begin{tikzpicture}[scale=0.3]
  \draw (1,0) rectangle (2,3.5);
  \draw (1,2.5) -- (2,2.5);
  \draw (1,2.5) -- (1,2.5);
  \draw (1,1) -- (2,1);
  \draw[snake=brace] (2.2,3.5) -- (2.2,0);
  \node at (1.5,1.5) {$\cdot$};
  \node at (1.5,2) {$\cdot$};
  \node at (3.3,1.8) {\mbox{\scriptsize $n$-$2$}};
\end{tikzpicture}}
= \bigoplus^{n-2}_{r=0}\,
\raisebox{-0.8cm}{\begin{tikzpicture}[scale=0.3]
  \draw[snake=brace] (-0.2,-2.2) -- (-0.2,3.5);
  \draw (0,-2.2) rectangle (1,3.5);
  \draw (1,0) rectangle (2,3.5);
  \draw (0,2.5) -- (2,2.5);
  \draw (1,1) -- (2,1);
  \draw (0,-1.2) -- (1,-1.2);
  \draw[snake=brace] (2.2,3.5) -- (2.2,0);
  \node at (-1.4,0.8) {\mbox{\scriptsize $\ell$ -$r$}};
  \node at (0.5,0.5) {$\cdot$};
  \node at (0.5,1) {$\cdot$};
  \node at (0.5,1.5) {$\cdot$};
  \node at (0.5,0) {$\cdot$};
  \node at (0.5,-0.5) {$\cdot$};
  \node at (0.5,2) {$\cdot$};
  \node at (1.5,1.5) {$\cdot$};
  \node at (1.5,2) {$\cdot$};
  \node at (2.8,1.7) {\mbox{\scriptsize$r$}};
\end{tikzpicture}}\; ,
\end{equation}
where the dimensions of those at the $r.h.s.$ are
\begin{equation}
     \left(\begin{array}{c} \mathrm{dim}\fG+1\\ r \end{array}\right)
     \left(\begin{array}{c} \mathrm{dim}\fG\\ \ell-r \end{array}\right)
 \frac{\ell-2r+1}{\ell-r+1}\;  ;
\end{equation}
note that in all terms above $\ell-r\leq \mathrm{dim}\fG$. Since the
first (longest) column of the Young patterns above may have
dim$\fG$ boxes at the most, $2n-3-r\leq \mathrm{dim}\fG$ or
$\ell-r\leq \mathrm{dim}\fG$ .
The above decomposition of the outer product at the $l.h.s.$ of
\eqref{decomp} in representations of the $S_{2n-3}$ symmetric group
determines the possible `elementary' or
($GL(\mathrm{dim}\fG)$)-irreducible symmetry patterns of the
$\ell$-bracket. In particular, the $r=0$ component above would
correspond to a fully skewsymmetric ($2n-3$)-bracket and hence to
a ($2n-3$)-Lie algebra. We now argue that in this context
the adequate generalization of the Lie triple system
requires that the bracket of the Lie $\ell$-ple system
has the symmetry of the $r=n-2$ Young pattern in the sum \eqref{decomp}
{\it i.e.}, that the bracket has the symmetry
determined by
\raisebox{-0.6cm}{\begin{tikzpicture}[scale=0.3]
  \draw[snake=brace] (-0.2,-1) -- (-0.2,3.5);
  \draw (0,-1) rectangle (1,3.5);
  \draw (1,0) rectangle (2,3.5);
  \draw (0,2.5) -- (2,2.5);
  \draw (0,1) -- (2,1);
  \draw (0,0) -- (1,0);
  \draw[snake=brace] (2.2,3.5) -- (2.2,0);
  \node at (-1.3,1.2) {\mbox{\scriptsize $n$-$1$}};
  \node at (0.5,1.5) {$\cdot$};
  \node at (0.5,2) {$\cdot$};
  \node at (1.5,1.5) {$\cdot$};
  \node at (1.5,2) {$\cdot$};
  \node at (3.3,1.8) {\mbox{\scriptsize $n$-$2$}};
\end{tikzpicture}},
which indeed reduces to $\Yboxdim{8pt}\yng(2,1)$ for $n=3$.

To this aim, let us go back to the metric $\ell$-Leibniz algebra in Ex.~\ref{ht} as it
follows from Cor.~\ref{corcinco}. Since $\textrm{dim}\fG=n+1$, there is a restriction
since $\ell-r\leq n+1$. Thus, $r\geq n-4$ and, therefore, $n-4\leq r \leq n-2$.
Consider now the $(2n-3)$-commutators as defined by \eqref{2n-3-brack}
\begin{equation}
\label{l-ple}
[\mathbf{e}_{a_1},\dots,\mathbf{e}_{a_{n-1}},\mathbf{e}_{b_1},\dots,\mathbf{e}_{b_{n-2}}] =
-\sum_{\sigma\in
    S_{n-1}} \delta_{a_{\sigma(1)b_1}} \dots \delta_{a_{\sigma(n-2)} b_{n-2}}
    \delta_{a_{\sigma(n-1)}}{}^c \mathbf{e}_c  \; .
\end{equation}
To see how \eqref{l-ple} selects a specific symmetry among
the
irreducible components in the $r.h.s.$ of \eqref{decomp}, let
us look at the symmetry of a generic pattern. The $2n-3$ indices
of the Young tableau are split into two sets with
$2n-3-r$ and $r$ indices respectively.
The primitive projector associated to the Young tableau
symmetrizes $r$ pairs of indices, where each pair
contains an index of the first set and another of the second
one and, then, it antisymmetrizes the indices of both sets separately.
This projector, applied to the $r.h.s.$ of
eq.~\eqref{l-ple}, gives zero due to the $\delta_{ab}$ factors
unless the indices of the first set are the ($n-1$) $a$'s
all placed in the first column of the Young tableau
(and thus the indices of the second set are
the $n-2$ $b$'s in the second column),
since otherwise there will be a $\delta$ with
antisymmetrized indices. Thus, $\ell-r=n-1$, $r=n-2$,
select in \eqref{decomp} the pattern that
corresponds to the $\ell$-bracket \eqref{l-ple}.


This motivates our second generalization of Lie triple systems:
\begin{definition}
A {\it Lie $\ell$-ple system}, $\ell=(2n-3)$, is a real
vector space $\mathfrak{L}$ endowed with a bracket given by a $\ell$-linear map
$\mathfrak{L}\times \mathop{\dots}\limits^{2n-3} \times
\mathfrak{L}\rightarrow \mathfrak{L}$, $(X_1,\dots , X_{2n-3})
\mapsto [X_1,\dots , X_{2n-3}]$ such that
\begin{enumerate}
    \item  it is antisymmetric in the first $n-1$ and
    in the last $n-2$ indices;
    \item  it satisfies the (left) FI;
    \item  its overall symmetry structure is given by the Young pattern\\
  \raisebox{-0.6cm}{\begin{tikzpicture}[scale=0.3]
  \draw[snake=brace] (-0.2,-1) -- (-0.2,3.5);
  \draw (0,-1) rectangle (1,3.5);
  \draw (1,0) rectangle (2,3.5);
  \draw (0,2.5) -- (2,2.5);
  \draw (0,1) -- (2,1);
  \draw (0,0) -- (1,0);
  \draw[snake=brace] (2.2,3.5) -- (2.2,0);
  \node at (-1.3,1.2) {\mbox{\scriptsize $n$-$1$}};
  \node at (0.5,1.5) {$\cdot$};
  \node at (0.5,2) {$\cdot$};
  \node at (1.5,1.5) {$\cdot$};
  \node at (1.5,2) {$\cdot$};
  \node at (3.3,1.8) {\mbox{\scriptsize $n$-$2$}};
\end{tikzpicture}}
\end{enumerate}
\end{definition}
\noindent
Properties (a) and (b) above define a particular
(left) $\ell$-Leibniz algebra structure $\fL$, $\ell$ odd;
(c) makes of $\fL$ an $\ell$-ple system.
Note that, strictly speaking, (a) above is included in (c)
and that, due to the properties of the projectors that determine the
($GL(\mathrm{dim}\fL)$-) irreducible symmetries associated with the
different patterns, (c) automatically implies that the symmetrizations
and subsequent antisymmetrizations implied by any of the other
$(r\neq n-2)$ Young patterns in the $r.h.s.$ of \eqref{decomp} give zero
necessarily. When $n=3=l$, the resulting Lie triple system is
the standard one (Def.~\ref{triplesystems}).\\

\section{Concluding remarks  }
\label{f.r.}

In this paper we have introduced two Lie $\ell$-ple
generalizations of the Lie 3-ple, or triple, systems; they appear as special
cases of $k$-Leibniz algebras, themselves a generalization of
$k$-Lie algebras. As mentioned in the Introduction, 3-Lie algebras
underlie the BLG model; they are also behind the Basu-Harvey (BH)
\cite{Bas-Har:05} equation, which is naturally recovered as a
BPS condition of the BLG theory (the BH equation was originally
given in terms of a GLA four-bracket with a fixed entry;
see \cite{review} for $n$-Lie algebras given in terms of
($n+1$)-multibrackets of GLAs defined by the fully antisymmetrized
{\it associative} products of its entries). It is natural to think of physical
applications for larger $k>3$ algebras. In fact, higher order FAs appeared in
suspersymmetric physics before the advent of the BLG model:
their FIs may be thought of as generalized Pl\"ucker
relations \cite{FOF-Pap:02}, and these arise naturally
in the classification of maximally supersymmetric solutions of
supergravity theories. Thus, from this point of view, 4-Lie algebras
are relevant \cite{FOF-Pap:02} for the maximally supersymmetric
backgrounds in  IIB supergravity.

 Let us go back to the BH equation for M2 branes ending on a M5
brane. This relation may be considered as a generalization of the Nahm
equation \cite{Nahm:80} for D1 branes ending on a D3 brane, which
involves an ordinary Lie bracket. The Nahm and the BH BPS
equations have, respectively, the schematic form
$\dot{X}(s) \sim [X,X]\,,\, \dot{X}(s)\sim [X,X,X]$,
where $s$ is some direction in the D1 or M2 branes along which they
extend apart from the D3 and the M5 ones ({\it e.g.}, the D3 brane is
located at $s$=0). The structure of the Nahm and BH equations immediately
suggest moving to a general $n$-Lie bracket to write \cite{Bo-Ta-Za:08}
$\frac{dX}{ds}\sim [X,\mathop{\dots} \limits^n ,X]$;
in fact, all these expressions have the appearance of
Maurer-Cartan equations for FAs (see \cite{review}).

  To see the effect of a possible $n$-Lie generalization, let us
first recall how the $D$=11 M2-M5 system, with coordinates
\begin{equation}
\label{M2M5}
\begin{array}{cccccccc}
  M2: & 0 & 1 & 2 &   &  &   &   \\
  M5: & 0 & 1 &   & 3 & 4& 5 &  6 \quad ,
\end{array}
\end{equation}
is described. From the M2 worldvolume point of view, the M5 brane is given
by four transverse 3-Lie algebra-valued functions
$X^{\mathcal{I}}(s),\ \mathcal{I}=3,4,5,6$,
where $s$ corresponds to the spatial M2 worldvolume coordinate transverse
to the M5 brane (the second one), and which obey the BH equation
for a 3-Lie bracket. From the (dual) point of view of the M5 brane, the
coordinate $s$ becomes a field depending on the transverse $X^{\mathcal{I}}$
coordinates of the M5 brane, $s=s(X^3,X^4,X^5,X^6)$. In the general
case, $\frac{dX}{ds}\sim [X,\mathop{\dots} \limits^n ,X]$,
we may think of a generic solution with the behaviour
$X(s)\sim \frac{1}{s^{\frac{1}{n-1}}}$, $s\sim \frac{1}{X^{n-1}}$,
where the exponent is determined by the number $n$ of entries
of the $n$-bracket. We would expect $s$, as field, to be a harmonic
function in $d$-dimensional transverse space,
for which we would need $n=d-1$ since a harmonic function in
$d$-dimensions depends on the radius as $1/R^{d-2}$.
Thus, the Nahm  (BH) equations correspond
to $n=2\ (3)$ since in the D1-D3 (M2-M5) systems the D1 (M2) branes
appear, from the point of view of the D3 (M5) ones, as a scalar
field with the behaviour $s\sim \frac{1}{R}$ ($s\sim \frac{1}{R^2}$).
Thus we may speculate, for {\it e.g.}  $D$=10, whether
other (supersymmetric) brane systems determined by suitable
`brane-boundary rules' \cite{Stro:95, Town:96} could be described
by a generalized BH equation involving other $n$-Lie algebra
brackets (we thank Neil Lambert on this point).

   Further, there is also the question of moving from $n$-Lie
to the more general $n$-Leibniz algebras with non fully
anticommuting brackets; in particular, $n$-Leibniz algebras which
retain fully skewsymetric fundamental objects appear often as an important
subclass (see \cite{review,deAz-JMI:10}), as we have also seen in this paper. It turns
out here that there is also a BPS relation \cite{Pal-Sa:11} that is the
BH-like equation for a $(\mathcal{N}=2)$-supersymmetric BLG-type model
\cite{Cher-Sa:08}, which uses CS algebras rather than $3$-$Lie$ ones.
Thus, all the above considerations provide a motivation for
considering, setting aside their mathematical interest, the various
higher order $k$-Leibniz and, in particular, Lie $k$-ple algebras
introduced here, and raises the issue of their
possible applications.

\vskip 1cm
\noindent
{\bf Acknowledgements}.
The authors wish to thank Neil Lambert for a helpful conversation.
This work has been partially supported by research grants from the
Spanish MINECO (FIS2008-01980, FIS2009-09002, CONSOLIDER
CPAN-CSD2007-00042).


\begin{thebibliography}{99}


\bibitem{Nambu:73}
Y.~Nambu,
{\it Generalized Hamiltonian dynamics},
Phys. Rev. {\bf D7}, 2405-2414 (1973).

\bibitem{AzPePB:96a}
J.~A. de~Azc\'arraga, A.~M. Perelomov, and J.~C. P\'erez~Bueno,
{\it New Generalized Poisson Structures}, J. Phys. {\bf A29},
L151-L157 (1996), [arXiv:q-alg/9601007]; {\it The
Schouten-Nijenhuis bracket, cohomology and generalized Poisson structures},''
J. Phys. {\bf A29}, 7993-8010 (1996)  [arXiv:hep-th/9605067].

\bibitem{AzBu:96}
J.~A. de~Azc\'arraga and J.~C. P\'erez-Bueno,
{\it Higher-order simple Lie algebras},
Commun. Math. Phys. {\bf 184}, 669-681 (1997),
[arXiv:hep-th/9605213].

\bibitem{Han-Wac:95}
P. Hanlon and H. Wachs, {\it On Lie $k$-algebras}, Adv. in Math.
{\bf 113}, 206-236 (1995).

\bibitem{Gne:95}
V. Gnedbaye, {\it Les alg\`ebres $k$-aires et leurs op\'erads},
C. R. Acad. Sci. Paris, S\'erie I, {\bf 321} (1995) 147-152.

\bibitem{Vin2:98}
A.~M. Vinogradov and M.~M. Vinogradov, {\it On multiple generalizations of {L}ie
  algebras and Poisson manifolds},  Contemp. Math. {\bf 219},
  273-287 (1998)

\bibitem{Filippov}
V.~Filippov, {\it $n$-{L}ie algebras},
Sibirsk. Mat. Zh. {\bf 26} (1985), no.~6, 126-140, (1985)
[Engl. trans.: Siberian Math. J. \textbf{26}, no.~6, 879-891 (1985)].

\bibitem{Kas:87}
S.~M. Kasymov, {\it Theory of $n$-lie algebras},
Algebra i Logika {\bf 26}, no.~3, 277-297 (1987)
[Engl. trans.: Algebra and Logic, {\bf 26}, 155-166 (1988)].

\bibitem{Ling:93}
W.~X. Ling, {\it On the structure of $n$-Lie algebras}.
PhD thesis, Siegen, 1993.

\bibitem{Cas-Lod-Pir:02}
J.~Casas, J.-L. Loday, and T.~Pirashvili,
{\it Leibniz $n$-algebras},
Forum Math. {\bf 14}, 189-207 (2002);\\
J.-L. Loday,
{\it Une version non-commutative des alg\`ebres de Lie},
L'Ens. Math. {\bf 39}, 269-293 (1993).

\bibitem{Tak:93}
L.~Takhtajan, {\it On Foundation of the generalized Nambu mechanics},
Commun. Math. Phys. {\bf 160}, 295-316 (1994)
[arXiv:hep-th/9301111].

\bibitem{Sah-Val:92}
D.~Sahoo and M.~C. Valsakumar,
{\it Nambu mechanics and its quantization},
Phys. Rev. {\bf A46}, 4410-4412 (1992).

\bibitem{AIP-B:97}
J.~A. de~Azc\'arraga, J.~M. Izquierdo, and J.~C. P\'erez~Bueno,
{\it On the higher-order generalizations of Poisson structures},
J. Phys. {\bf A30}, L607-L616 (1997)
[arXiv:hep-th/9703019].

\bibitem{Cu-Za:02}
T.~Curtright and C.~K. Zachos,
{\it Classical and quantum Nambu mechanics},
Phys. Rev. {\bf D68}, 085001 (2003)
[arXiv:hep-th/0212267].

\bibitem{review}
J.~A.~de Azc\'arraga and J.~M.~Izquierdo,
{\it n-ary algebras: a review with applications},
J. Phys. {\bf A43} (2010) 293001-1-117
[arXiv:1005.1028 [math-ph]];
{\it Topics on n-ary algebras},
J. Phys. Conf. Ser. {\bf 284}, 012019 (2011)
[arXiv:1102.4194 [math-ph]].

\bibitem{Ba-La:06}
J.~Bagger and N.~Lambert, {\it Modeling multiple M2's},
Phys. Rev. {\bf D75} 045020  (2007),
[arXiv:hep-th/0611108];
{\it Gauge symmetry and supersymmetry of multiple M2-branes},
Phys. Rev. {\bf D77} (2008) 065008
[arXiv:0711.0955 [hep-th]].

\bibitem{Gust:08a}
A.~Gustavsson, {\it Selfdual strings and loop space Nahm
equations}, {\em JHEP} {\bf 04},  083 (2008),
[arXiv:0802.3456 [hep-th]].

\bibitem{Aha-Be-Ja-Mal:08}
O.~Aharony, O.~Bergman, D.~L. Jafferis, and J.~Maldacena, {\it N=6 superconformal
  Chern-Simons-matter theories, M2-branes and their gravity duals},
  JHEP {\bf 10}, 091 (2008)
  [arXiv:0806.1218 [hep-th]].

\bibitem{Ba-La-Mu-Pa:12}
J. Bagger, N. Lambert, S. Mukhi and C.~Papageorgakis,
{\it Multiple membranes in M-theory},
arXiv:1203.3546 [hep-th].

\bibitem{Pap:08}
G.~Papadopoulos, {\it M2-branes, 3-Lie Algebras and Pl\"ucker
relations}, JHEP {\bf 05},  054 (2008)
[arXiv:0804.2662 [hep-th]].

\bibitem{Ga-Gu:08}
J.~P. Gauntlett and J.~B. Gutowski,
{\it Constraining maximally supersymmetric membrane actions},
JHEP {\bf 0806}, 053 (2008)
[arXiv:0804.3078 [hep-th]].

\bibitem{Go-Mi-Ru:08}
J.~Gomis, G.~Milanesi, and J.~G. Russo,
{\it Bagger-Lambert theory for general  Lie algebras},
JHEP {\bf 06}, 075 (2008)
[arXiv:0805.1012 [hep-th]].

\bibitem{Cher-Sa:08}
S.~Cherkis and C.~S\"amann,
{\it Multiple M2-branes and generalized 3-Lie algebras}
Phys. Rev. {\bf D78} 066019 (2008),
{\tt [arXiv:0807.0808 [hep-th]]};\\
S.~Cherkis, V.~Dotsenko and C.~S\"amann,
{\it On superspace actions for multiple M2-branes, metric 3-algebras and their classification},
Phys.\ Rev.\ D {\bf 79}, 086002 (2009)
[arXiv:0812.3127 [hep-th]].

\bibitem{Ba-La:08}
J.~Bagger and N.~Lambert,
{\it Three-algebras and N=6 Chern-Simons gauge Theories}
Phys. Rev. {\bf D79}, 025002 (2009)
[arXiv:0807.0163 [hep-th]].

\bibitem{Jacob:49}
N.~Jacobson, {\it Lie and Jordan triple systems},
Amer. J. Math. {\bf 71}, 149-170 (1949);
{\it General representation theory of Jordan algebras},
Trans. Amer. Math. Soc. {\bf 70}, 509-530 (1951).

\bibitem{Lis:52}
W.~G. Lister,
{\it A structure theory of Lie triple systems},
Trans. Am. Math. Soc. {\bf 72}, 217-242 (1952).

\bibitem{Yama:57}
K.~Yamaguti,
{\it On algebras of totally geodesic spaces (Lie triple systems)},
J. Sci. Hiroshima Univ. Ser. A {\bf 21},  107-113, (1957);
{\it On the Lie triple system and its generalization}, {\it ibid}
{\bf 21}, 155-160 (1958).

\bibitem{Fa:73}
J. R. Faulkner, {\it On the geometry of inner ideals},
J. of Algebra {\bf 26}, 1-9 (1973).

\bibitem{Ber:00}
W.~Bertram,
{\it The geometry of Jordan and Lie structures},
Springer Lecture Notes in Mathematics {\bf 1754}, Berlin, 2000.

\bibitem{Bre-San:11}
M.~R.~Bremner and J.~S\'anchez-Ortega,
{\it Leibniz triple systems},
arXiv:1106.5033 [math.RA].

\bibitem{Gau:96}
P.~Gautheron,
{\it Some remarks concerning Nambu mechanics},
Lett. Math. Phys. {\bf 37}, 103-116 (1996).

\bibitem{Kas:95a}
S.~M. Kasymov,
{\it Analogs of the Cartan criteria for $n$-Lie algebras},
Algebra i Logika {\bf 34}, no.~3, 274-287 (1995)
[Engl. trans.: Algebra and Logic \textbf{34}, no.~3, 147-154 (1995)].

\bibitem{Da-Tak:97}
Y.~L. Daletskii and L.~Takhtajan,
{\it Leibniz and Lie algebra structures for Nambu algebra},
Lett. Math. Phys. {\bf 39}, 127-141 (1997).

\bibitem{Okub:94}
S.~Okubo,
{\it Parastatistics as Lie supertriple systems},
J. Math. Phys. {\bf 35}, 2785-2803 (1994)
[arXiv:hep-th/9312180].

\bibitem{Okubo:93}
S.~Okubo,
{\it Triple products and Yang-Baxter equation:
I, Octonionic and quaternionic triple systems;
II, Orthogonal and symplectic ternary systems},
J. Math. Phys. {\bf 34}, 3273-3291; {\it ibid.} 3292-3315 (1993)
[arXiv:hep-th/9212052].

\bibitem{Oku-Kam:96a}
S.~Okubo and N.~Kamiya,
{\it Quasi-classical Lie superalgebra and Lie-super triple systems}
[q-alg/9602037].

\bibitem{Ker:00}
R. Kerner,
{\it Ternary algebraic structures and their applications in physics},
math-ph/0011023, Proc. of the XXIII ICGTMP, Dubna (2000)

\bibitem{Oku-Kam:96b}
S.~Okubo and N.~Kamiya,
{\it Jordan-Lie superalgebras and Jordan-Lie triple systems}
J. Alg. {\bf 398}, 388-411 (1997), UR-1467.

\bibitem{Oku:03}
S.~Okubo,
{\it Construction of Lie superalgebras from triple product systems}
AIP Conf. Proc. {\bf 687}, 33-40 (2003)
[math-ph/0306029].

\bibitem{deMed-JMF-Men-Rit:08}
P.~de~Medeiros, J.~Figueroa-O'Farrill, E.~M\'endez-Escobar, and
P.~Ritter,
{\it On the Lie-algebraic origin of metric 3-algebras},
Commun. Math. Phys. {\bf 290} 871-902 (2009)
[arXiv:0809.1086 [hep-th]].

\bibitem{Palm:09}
J.~Palmkvist,
{\it Three-algebras, triple systems and 3-graded Lie superalgebras},
J. Phys. {\bf A43} (2010) 015205
[arXiv:0905.2468 [hep-th]].

\bibitem{Bas-Har:05}
A.~Basu and J.~A. Harvey, {\it The M2-M5 brane system and a generalized Nahm's
equation}, Nucl. Phys. {\bf B713}, 136-150 (2005)
[arXiv:hep-th/0412310].

\bibitem{FOF-Pap:02}
J.~M.~Figueroa-O'Farrill and G.~Papadopoulos,
{\it Pl\"ucker type relations for orthogonal planes},
J. Geom. and Phys. {\bf 49}, 294-331 (2004)
[math/0211170 [math-ag]].

\bibitem{Nahm:80}
W.~Nahm, {\it A Simple Formalism for the BPS Monopole},
Phys. Lett. {\bf B90}, 413-414  (1980)

\bibitem{Bo-Ta-Za:08}
G.~Bonelli, A.~Tanzini and M.~Zabzine,
{\it Topological branes, p-algebras and generalized Nahm equations},
Phys. Lett. {\bf B672}, 390-395 (2009)
[arXiv:0807.5113 [hep-th]].

\bibitem{Stro:95}
A.~Strominger,
{\it Open p-branes},
Phys. Lett. {\bf B383}, 44-47 (1996)
[hep-th/9512059].

\bibitem{Town:96}
P.~K.~Townsend,
{\it Brane surgery},
Nucl. Phys. Proc. Suppl. {\bf 58}, 163-175 (1997)
[hep-th/9609217].

\bibitem{deAz-JMI:10}
J.~A.~de Azc\'arraga and J.~M.~Izquierdo,
{\it On a class of n-Leibniz deformations of the simple Filippov algebras},
J.\ Math.\ Phys.\  {\bf 52}, 023521 (2011)
[arXiv:1009.2709 [math-ph]].

\bibitem{Pal-Sa:11}
S.~Palmer and C.~S\"amann,
{\it Constructing generalized self-dual strings},
JHEP {\bf 1110}, 008 (2011)
[arXiv:1105.3904 [hep-th]].

\end{thebibliography}
\end{document}